\documentclass[letterpaper, reqno,  11pt]{amsart}

\usepackage[T1]{fontenc} \usepackage{lmodern,times}
\usepackage{microtype}

\usepackage{comment,url,xcolor, amsmath, amsthm, amssymb}

\usepackage{enumitem}
\setlist[enumerate]{itemsep=0.7ex,leftmargin=7ex}
\newlist{asmlistA}{enumerate}{1}
\setlist[asmlistA]{label*=(\textbf{A\arabic*}),itemsep=0.7ex, resume}
\newlist{asmlistB}{enumerate}{1}
\setlist[asmlistB]{label*=(\textbf{B\arabic*}),itemsep=0.7ex, resume}

\newcommand{\pds}[1]{\frac{\partial}{\partial #1}}
\newcommand\bE{{\mathbb E}}
\newcommand{\toweak}{\stackrel{w}{\rightarrow}}
	
\newcommand\sN{{\mathcal N}}
\DeclareMathOperator\Int{Int}
\DeclareMathOperator\dom{dom}
\DeclareMathOperator*\esssup{esssup}

\newcommand{\tot}{\tfrac{1}{2}} 
\newcommand{\oo}[1]{\tfrac{1}{#1}}
\newcommand{\abs}[1]{\left| #1 \right|} 
\newcommand{\babs}[1]{\big| #1 \big|} 
\newcommand{\set}[1]{\{#1\}} 
\newcommand{\sets}[2]{\set{#1\,:\,#2}} 
\newcommand{\Bset}[1]{\Big\{#1\Big\}} 
\newcommand{\Bsets}[2]{\Bset{#1\,:\,#2}} 
\newcommand{\seq}[1]{\set{#1_n}_{n\in\N}} 
\renewcommand{\fam}[2]{\{ #1 \}_{#2}}

\newcommand{\prf}[1]{ \{ #1 \}_{t\in [0,T]}} 
\newcommand{\prfo}[1]{\{#1\}_{t\in [0,1]}}
\newcommand{\dd}{d}

\newcommand{\tRN}[2]{\tfrac{\dd #1}{\dd #2}}
\newcommand{\downto}{\searrow}
\providecommand{\R}{} \renewcommand{\R}{{\mathbb R}}

\providecommand{\N}{} \renewcommand{\N}{{\mathbb N}}
\newcommand{\PP}{{\mathbb P}}
\newcommand{\QQ}{{\mathbb Q}}

\newcommand{\EE}{{\mathbb E}}

\newcommand{\ee}[1]{ \bE \left[ #1 \right] }
\newcommand{\bee}[1]{ \bE \big[ #1 \big] }
\newcommand{\Bee}[1]{ \bE \Big[ #1 \Big] }
\newcommand{\eeq}[2]{ \bE^{#1}\left[ #2 \right] }
\newcommand{\EN}{{\mathcal E}}
\newcommand{\ba}{{\mathrm{ba}}}
\newcommand{\eps}{\varepsilon}
\newcommand{\vp}{\varphi}
\newcommand{\el}{{\mathbb L}} 
\newcommand{\lone}{\el^1}
\newcommand{\linf}{\el^{\infty}}
\newcommand{\define}[1]{{\textbf{#1}}}
\newcommand{\efor}{\text{ for }}
\newcommand{\eand}{\text{ and }}
\newcommand{\ewhere}{\text{ where }}
\newcommand\tW{{\tilde{W}}}
\newcommand\sZ{{\mathcal Z}}
\newcommand\hZ{{\hat{Z}}}
\newcommand{\FFF}{{\mathbb F}}

\newcommand{\fF}{{\mathfrak F}}
\newcommand{\fE}{{\mathfrak E}}
\renewcommand{\fE}{E}
\newcommand{\fS}{{\mathfrak S}}

\newcommand{\de}{D_{\fE}}

\newcommand{\tex}{(T,\eta,x)}

\newcommand{\peta}{\PP^{\eta}}
\newcommand{\etaine}{\eta\in \fE}
\newcommand{\petaeta}{(\peta)_{\etaine}}

\newcommand{\uV}{\underline{V}}
\newcommand{\uv}{\underline{v}}

\newcommand{\as}{\text{ a.s.}}

\def\PP{\mathbb P} 
\def\R{\mathbb{R}}

\def\sF{{\mathcal F}}
\def\sM{{\mathcal M}}
\def\sG{{\mathcal G}}
\def\sA{{\mathcal A}}
\def\sL{{\mathcal L}}

\def\sD{{\mathcal D}}

\def\sX{{\mathcal X}}

\def\sS{{\mathcal S}}

\def\sC{{\mathcal C}}

\def\E{\mathbb{E}}

\def\sF{\mathcal{F}}
\def\P{\mathbb{P}}

\def\Q{\mathbb{Q}}
\def\ld{\lambda} 
\def\sE{{\mathcal E}}
\def\N{\mathbb N}

\newcommand{\ips}{\textstyle\int_0^T \pi_u\, dS_u}
\newcommand{\uvp}{\underline{\vp}}
\newcommand{\hnu}{\hat{\nu}}
\newcommand{\smes}{\sM^{e}_{\sigma}}

\numberwithin{equation}{section}
\theoremstyle{plain}                
\newtheorem{theorem}{Theorem}[section]
\newtheorem{lemma}[theorem]{Lemma}
\newtheorem{proposition}[theorem]{Proposition}
\newtheorem{corollary}[theorem]{Corollary}
\theoremstyle{definition}           

\theoremstyle{remark}
\newtheorem{remark}[theorem]{Remark}

\begin{document}

\title{Facelifting in utility maximization}

\author{Kasper Larsen}
\address{Kasper Larsen, Department of Mathematical Sciences, Carnegie Mellon University}
\email{kasperl@andrew.cmu.edu}
\thanks{
  The authors would like to thank Mihai S\^\i rbu and Kimberly Weston for
  numerous and helpful discussions.  
  During the preparation of this work
  the second author has been supported by the Swiss National Foundation
  through the grant SNF $200021\_153555$ and by the Swiss Finance
  Institute.
  The third author has been supported by the National Science Foundation
  under Grants No.  DMS-0706947 (2010 - 2015) and Grant No.~DMS-1107465
  (2012 - 2017).  Any opinions, findings and conclusions or recommendations
  expressed in this material are those of the author(s) and do not
  necessarily reflect the views of the National Science Foundation (NSF).
}

\author{H.~Mete Soner}
\address{H.~Mete Soner, Department of Mathematics, ETH Z\" urich}
\email{mete.soner@math.ethz.ch}
 
\author{Gordan \v{Z}itkovi\'{c}}
\address{Gordan \v Zitkovi\' c, Department of Mathematics,  
University of Texas at Austin}
\email{gordanz@math.utexas.edu}



\begin{abstract} 
We establish the existence and characterization of a primal and a dual
facelift - discontinuity of the value function at the terminal time -  for
utility-maximization in incomplete semimartingale-driven financial markets.
Unlike in the lower- and upper-hedging problems, and somewhat unexpectedly,
a facelift turns out to exist in utility-maximization despite strict
convexity in the objective function. In addition to discussing our results
in their natural, Markovian environment, we also use them to show that the
dual optimizer cannot be found in the set of countably-additive
(martingale) measures in a wide variety of situations.  
\end{abstract}

\maketitle

\section{Introduction}
\label{sec:intro}

Valuation, or pricing, is one of the central problems in mathematical
finance. Its goal is to assign a dollar value to a contingent claim based
on the economic principles of supply and demand. When the claim is liquidly
traded in a financial market, the only meaningful notion of price is the
one at which the claim is trading. When the claim is not traded, but is
replicable in an arbitrage-free market, its unique price is determined by
the no-arbitrage principle. The case of a nonreplicable claim is the most
complex one. Here, the no-arbitrage principle alone does not suffice and
additional economic input is needed.  This input usually comes in the form
of a risk profile of the agents involved in the transaction. In extreme
cases, one comes up with the notions of upper and lower hedging prices,
while in between those the pricing procedure typically involves a solution
of a utility-maximization problem.

\medskip

Utility-maximization problems arise in other contexts, as well - in optimal
investment and equilibrium problems, e.g. In fact, they play a central role
in mathematical finance and financial economics. This fact is quite evident
from the range of literature both in mathematics, as well as economics and
finance,  that treats them. Instead of providing a list of the most
important references, we simply point the reader to the monograph
\cite{KarShr98} and the references therein for a thorough literature review
from the inception of the subject to 1998. The more recent history,
at least as far as the relevance to the present paper is
concerned, can be found in the papers \cite{KraSch99} and
\cite{CviSchWan01}, where the problem is treated in great mathematical
generality. 

\medskip

Both in pricing and utility-maximization, there is a significant jump in
mathematical and conceptual difficulty as one transitions from complete to
incomplete models. In pricing, it is well known that the upper (and lower)
hedging price of a nonreplicable claim cannot be expressed as the
expectation under a (local, $\sigma$-) martingale measure (see Theorem
5.16, p.~248 in \cite{DelSch98}).  In other words, when viewed as a linear
optimization problem over the set of martingale measures, the value of the
upper hedging problem is attained only when a suitable relaxation is
introduced. This relaxation almost always (implicitly or explicitly)
involves a closure in the \mbox{weak-$*$} topology and the passage from
countably-additive to merely finitely-additive measures. This phenomenon is
well understood not only from the functional-analytic, but also from the
control-theoretic and analytic points of view.  Indeed, 
stochastic target problems
(introduced in \cite{SonTou02}; see \cite{Tou13} for an overview and
further references) provide an approach using the related
partial differential equations. We also
understand that the passage from countable to finite additivity
corresponds, loosely speaking, to the lack of weak compactness of
minimizing sequences,  and that it often corresponds to a discontinuity in
the problem's value function at the terminal time.  In mathematical finance
this naturally leads to a ``facelifting'' procedure where one upper- (or
lower-) hedges a contingent claim by (perfectly) hedging another contingent
claim whose payoff is an upper majorant of the original payoff in a
specific class (see, e.g.,  \cite{BroCviSon98}, \cite{BouTou00},
\cite{SchShrWys00}, \cite{SonTou00}, \cite{SonTou02}, and
\cite{GuaRasSch08}).

\medskip

The literature on nonlinear problems such as utility maximization and
utility-based pricing is much narrower in scope.  In this context one must
also distinguish between the need for relaxation and the existence of a
facelift. While, as we show in this paper, the existence of a facelift is
related to the non-existence of a minimizer without an appropriate
relaxation in most cases of interest, the opposite  implication does not
hold. In fact, the only known cases in the literature where it is shown
that a relaxation is necessary (in \cite{KraSch99} and \cite{HugKraSch05})
do not come with a facelift (as can be deduced from our main theorem).
Moreover, they appear in non-Markovian settings, are constructed using
heavy functional-analytic machinery (the Rosenthal's subsequence
splitting lemma in \cite{HugKraSch05}, e.g.), and do not involve a
nonreplicable random endowment.

When a nonreplicable random endowment is present - which is invariably
the case when one wants to consider pricing approaches other than marginal
utility-based pricing (such as indifference or conditional marginal pricing) - no answer can
be found in the existing literature.  Indeed, while the papers
\cite{CviSchWan01}, \cite{KarZit03} and others treat such problems
theoretically, and both pose and solve a class of dual utility-maximization
problems over an appropriate relaxation of the set of ($\sigma$- or local)
martingale measures, a proof of necessity of such an enlargement is never
given. 

One can speculate that one reason why such results do not exist is because
we never expected to see a facelift in such problems. Therefore, by a
somewhat perverted logic, we did not expect a finitely-additive
relaxation to be truly necessary, except in pathological cases. 
After all, the objective function is strictly convex - there are no
``flat parts'' to produce infinite Hamiltonians and the related explosion
in control which leads to the emergence of a facelift (see, e.g., \cite{Pha09},
Subsection 4.3.2, p.~69 for an accessible treatment). Indeed, if one tries
to apply the ``exploding Hamiltonian'' test to virtually any Markovian
incarnation of a (primal or dual) utility-maximization problem with a
random endowment, the results will be inconclusive - the Hamiltonian never
explodes. 
It came, consequently, as a great surprise to us when we discovered that
that the ``Hamiltonian test'' is impotent in this case and that the
facelift  appears virtually generically. Moreover, there is no need
for pathology at all. As we explain in our illustrative Section
  \ref{sec:illustrative}, in what one can quite confidently call the
  ``simplest nontrivial incomplete utility-maximization problem with
  nonreplicable random endowment'', the facelift invariably appears.
  Moreover, in many setups, 
  every time it appears, one can show that the corresponding dual
  problem does not admit a minimizer in the class of countably-additive
  measures. This fact is not only of theoretical value - it has important
  implications for the numerical treatment of the problem.

\medskip

After the aforementioned illustrative example in Section
\ref{sec:illustrative}, we turn to a general semimartingale model of a
financial market in Section \ref{sec:General} and analyze the asymptotic
behavior of the value function of the dual utility-maximization problem
with random endowment as the time-horizon shrinks to $0$. 
While keeping the
same underlying market structure, we let the random endowment vary with the
horizon in a rather general fashion. We show here that the limiting value
of the value function exists under minimal conditions on the inputs,
compute its value explicitly, and argue that it often differs from the
limiting value of the objective, i.e., that a facelift exists.

The choice of the shrinking time horizon - as opposed to the one of Section
\ref{sec:illustrative}, where the current time gets closer and closer to
the horizon - is made here for mathematical convenience. While it may be of
interest in it own right when one wants to study utility-maximization on very short horizons,
our main concern is to understand how the value function of the (dual)
utility-maximization problem behaves close to maturity. In Markovian
models, as described in Subsection \ref{sse:markov}, the two views can be
reconciled by observing that various control problems corresponding to
the same value of the state variable, but varying values of the time
parameter, can be coupled on the same probability space. This way, the study of 
the ``forward'' convergence of value functions can be aided by the natural RCLL
properties of trajectories of canonical Markov processes, and the abstract
results of Section \ref{sec:General}. 

In Section \ref{sec:mod}, we take up a related problem and show that under
mild conditions on the random endowment, the objective function in the dual
utility-maximization problem can be replaced by a smaller function without
 changing its value. This can be interpreted as the long-distance
incarnation of the facelift and we use it to show that 
if the random endowment is nonreplicable and its negative
admits a unique minimal (smallest) replicable majorant, then the dual
utility-maximization problem cannot have a solution among the
countably-additive measures. 

Section \ref{sec:Phi} is devoted to an in-depth study of the only
non-standard assumption made in our main theorem in Section
\ref{sec:General} - namely, the existence of the so-called germ price.
Therein, 
two general sufficient conditions are given and 
concrete examples where they hold are 
described.


\section{An illustrative example}
\label{sec:illustrative}

Before we develop a theory in a general semimartingale market model, the purpose of this section is to show that a facelift - together with all
of its repercussions such as nonattainment in the class of
countably-additive martingale measures - already appears in the simplest of models and is not a ``cooked-up'' consequence of a pathological choice of the modeling framework. 

\subsection{The market model}
On a given  time horizon $T>0$, we let $\prf{B_t}$ and $\prf{W_t}$ be two
independent Brownian motions, and $\FFF=\prf{\sF_t}$ the standard
augmentation of their natural filtration $\FFF^{B,W}$.  The financial
market model consists of a \define{money-market account} $\prf{S^{(0)}_t}$ and a
\define{risky security} $\prf{S_t}$. For simplicity, we assume a zero
interest rate, i.e., $S^{(0)} \equiv 1$, and we model $S$ by the geometric
Brownian motion: \begin{align}\label{def:dS_gbb} dS_t := S_t \big( \mu\,  dt
+  \sigma\, dB_t\big),\quad S_0 :=1.   \end{align} 
Assuming throughout that $\sigma > 0$ and $\mu\ne 0$, we set
$\ld=\mu/\sigma$ and interpret $\lambda$ as the \define{market price of
risk}.  So defined, our model follows completely the
Black-Scholes-Samuelson paradigm;  the Brownian motion $W$ will play a role
in the dynamics of the random endowment which we describe below.
\subsection{Trading and admissibility} The investor's initial wealth is
denoted by $x$;
at time $t\in[0,T]$ he/she holds $\pi_t$ shares of the stock $S$. The usual self-financing condition dictates that the
agent's total wealth admits the following dynamics 
\begin{align}\label{dX}
 X^{x,\pi}_t = x + \int_0^t S_
 u\pi_u \big(\mu\, du+ \sigma\, dB_u),\quad t\in[0,T].
 \end{align}
To ensure that the integral is well-defined we require 
that $\int_0^T \pi_u^2\,
du<\infty$, a.s. When, additionally, there exists a constant $a$ such that
$X^{x,\pi}_t \geq -a$, for all $t\in
[0,T]$, $\PP$-a.s.,  we call
$\pi$ \define{admissible} and we write $\pi\in \sA$. 

\subsection{Preferences and random endowment} 
For the purposes of this example, we model the agent's preferences by a
utility function of the ``power'' type, but also note that all of the statements
in this section remain true for a much larger class:
\[ U(x) := \tfrac{1}{p} x^{p}\  \efor\  p\in (-\infty,1)\setminus \set{0}\ 
\text{ or } \ 
U(x):=\log(x)\ \efor\  p=0.\] 
For definiteness, we set $U(x)=-\infty$ for $x<0$ (and at $x=0$ for $p\leq
0$). 

In addition to the investment opportunities provided by the financial
market $(S^{(0)},S)$, the investor receives a lump-sum payment (a random
endowment, stochastic income, etc.) at time $T$ of the form
$\varphi(\eta_T)$ where $\varphi:\R\to \R$ is a bounded continuous function
and $\eta_t:=\eta_0+W_t$, for $t\in [0,T]$. We note that the endowment
$\vp(\eta_T)$ cannot be replicated by trading in $(S^{(0)},S)$ as soon as
$\vp$ is not a constant function. On the other hand, more and more
information about its value is gathered by the agent as $t$ goes to $T$, so
it cannot be treated as an independent random variable, either. 

\subsection{The primal problem} 
Keeping track of the time horizon $T>0$, the initial wealth $x\in\R$ and
the initial value $\eta_0\in\R$ of the process $\eta$, we pose the
following optimization problem faced by a rational agent with
the characteristics described above: \begin{align}\label{CSW_primal}
u(T,\eta_0,x):= \sup_{\pi\in \sA} \Bee{ U\Big(X^{x,\pi}_T +
\vp(\eta_T)\Big) }. \end{align}
Let $x_c=x_c(T,\eta_0)\in\R$ be such that $u(T,\eta_0,x)=-\infty$ for
$x<x_c$ and $u(T,\eta_0,x)>-\infty$ for $x>x_c$. In \cite{CviSchWan01} it is
shown that
$-x_c$ coincides with the superreplication cost of $-\vp(\eta_T)$. In our
case, thanks to the fact that $\eta_t =
\eta_0 +W_t$, we have $x_c = - \inf \vp$, independently of $\eta_0$ and
$T>0$. 

\subsection{The dual problem} Let $\sM$ denote the set of all 
$\P$-equivalent probability measures $\QQ$ on $\sF_T$ for which $S$ defined
by \eqref{def:dS_gbb} is a $\QQ$-martingale. In our, simple, model the
structure of $\sM$ is well known and completely described. Indeed, a
probability measure $\QQ$ is in $\sM$ if and only if $\Q$'s
$\sF_T$-Radon-Nikodym derivative is given by  
$\tRN{\QQ}{\PP}=Z_T$ where $Z$ is an exponential martingale of the
(differential) form
\begin{align}\label{dZ} dZ_t^\nu = -Z^\nu_t \Big( \ld\, dB_t + \nu_t
\, dW_t\Big),\quad Z_0=1, \end{align} for some progressive 
process $\nu$ with $\int_0^T \nu^2_u\, du<\infty$, a.s.

With the dual utility function given by  $V(z) := \sup_{x>0} \Big( U(x) -
xz\Big)$, for $z>0$, we define the 
value function of the \define{dual problem}
by
\begin{align}\label{equ:dual_value2}
v(T,\eta_0,z) := \inf_{\QQ\in\sM} \Big( 
 \bee{ V(z \tRN{\QQ}{\PP})} + z\eeq{\QQ}{ \vp(\eta_T)} \Big).
\end{align}
for $z>0$, $\eta\in\R$ and $T>0$.  
\begin{remark} 
\label{rem:infZ}
To guarantee the existence of a minimizer, 
in \cite{CviSchWan01} the authors identify 
$\sM$ with a subset of $\mathbb{L}_+^1(\PP)$ and embed it, naturally,
into the bi-dual
$\ba(\PP):=\linf(\PP)^{*} \supseteq \lone(\PP)$. 
With the weak$^*$-closure of $\sM$ in $\ba(\PP)$
is denoted by $\overline{\sM}_T$ and
the dual pairing
between $\mathbb{L}^\infty(\PP)$ and $\ba(\PP)$ 
by $\langle \cdot,\cdot\rangle$, the (relaxed) dual value function is defined
there by
\begin{align}\label{dual_value1}
\tilde{v}(T,\eta_0,z) := \inf_{\QQ\in\overline{\sM}_T} \left( 
 \ee{ V\left(z \tRN{\QQ^r}{\PP}\right)} + z\Big\langle \QQ,\vp(\eta_T)\Big\rangle \right),
\end{align} 
where $\QQ^r \in \mathbb{L}^1(\PP)$ denotes the regular part in the Yosida-Hewitt decomposition $\QQ = \QQ^r + \QQ^s$ of $\QQ\in\ba_+(\PP)$. 
It is shown in 
 \cite{CviSchWan01} that the dual
 minimizer $\hat{\QQ}= \hat{\QQ}^{T,\eta,z}$ is always attained in
 $\overline{\sM}_T$. 

Theorem 2.10, p.~675 in \cite{LarZit12} states that $v=\tilde{v}$, i.e.,
that the finitely-additive relaxation is unnecessary, if one is interested
in the value function alone.
 This allows us to work with random variables
 $\frac{d\QQ}{d\PP}\in\mathbb{L}^1_+(\PP)$ in the sequel, instead of 
 finitely-additive measures and their regular parts, needed in \eqref{dual_value1}. 
\end{remark}

\subsection{A na\"ive approach via HJB} If we were to approach
the utility-maximization problem via the formal dynamic
programming principle, we would start by embedding it into a family of problems starting
at $t\in [0,T]$, with the terminal time $T$, and depending additionally on the states
$x$ and $\eta$. Thanks to the Markovian structure, and without loss of
generality, instead of varying the initial time $t$, we use the
``time-to-go'' variable $T-t$. Moreover, we abuse the notation and denote
this variable simply by $T$, giving it an alternative interpretation of the
(varying) time horizon. This way, all the effects in the regime $t\sim T$
show up at $T\sim 0$.
The formal HJB equation is now given by: 
\begin{equation}
\label{equ:HJBu}
\left\{
\begin{aligned}
 u_T &= \sup_{\pi\in\R} \sL^{\pi} u\\
 u(0,\eta,x) &= U(x+\vp(\eta)).
 \end{aligned}
\right.
\end{equation}
Here $u_T:= \frac{\partial}{\partial T}u$  and $\sL^{\pi}$ is the (controlled) formal infinitesimal generator of the process $(X^{x,\pi},\eta)$. With $X^{x,\pi}$ defined by \eqref{dX} and $\eta^\eta_t := \eta +W_t$ this generator becomes 
\[ \sL^{\pi} u := \mu \pi u_x + \tot \sigma^2 \pi^2 u_{xx}
+ \tot u_{\eta\eta}.\] 
Similarly, the HJB equation for the dual value function formally reads as
follows
\begin{equation}
\label{equ:HJBv}
\left\{
\begin{aligned}
 v_T &= \inf_{\nu\in\R} \sN^{\nu} u\\
 v(0,\eta,z) &= V(z) + z \vp(\eta),
\end{aligned}
\right.
\end{equation}
where  the dynamics \eqref{dZ} for $Z^\nu$ produces the generator
\[ \sN^{\nu} v = \tot z^2 \Big(\ld^2 + \nu^2 \Big) v_{zz}
+ \tot v_{\eta\eta} - z \nu v_{z\eta}.\]
The seemingly natural choices for the primal domain $\sD_u$ and the
interpretation of the initial condition in \eqref{equ:HJBu} are
\begin{enumerate}[label*=\arabic*'.]
\item $\sD_u=\sets{\tex\in [0,\infty)\times \R \times \R}{
  x+\vp(\eta)>0}$, and
\item $\lim_{T\downarrow 0} u\tex= U\big(x+\vp(\eta)\big)$.
\end{enumerate}
Similarly, the dual domain $\sD_v$ and the initial condition for the dual
problem are expected to be
\begin{enumerate}[label*=\arabic*'.,resume]
\item $\sD_v=\sets{(T,\eta,z)\in [0,\infty)\times \R \times \R}{
  z>0}$, and
\item $\lim_{T\downarrow 0}v(T,\eta,z)= V(z)+ z \vp(\eta)$.
\end{enumerate}
It turns out, however, that \dots
\subsection{\dots the na\" ive approach is not always the right one}
In the remainder of the paper we will  show in much greater generality that the 
prescriptions 1'-4'~ above do not fully correspond to reality. Even in the
simple Black-Scholes-type model \eqref{def:dS_gbb} with utilities of power
type, the value
functions behave quite differently. If we set 
  \[ \dom(u):= \Int\sets{\tex\in [0,\infty)\times\R\times\R}{ u\tex \in
    (-\infty,\infty)},\]
    and
  \[ \dom(v):=
    \Int\sets{(T,\eta,z)\in [0,\infty)\times\R\times (0,\infty)}{ v(T,\eta,z) \in
    (-\infty,\infty)},\]
we have the following result (a special case of Theorem \ref{thm:main} below):
\begin{proposition}
\label{pro:fl-Brown}
In the setting of the current section, we have
\begin{enumerate}
\item $\dom(u)=\Bsets{\tex\in (0,\infty)\times \R \times \R}{
  x>-\inf\vp}$.
\item 
$\lim_{T\downarrow 0} u\tex = \begin{cases} U(x+\vp(\eta)), & x\geq -
\inf \vp.\\
-\infty, & x<-\inf\vp. \end{cases}$
\item $\dom(v)=\Bsets{(T,\eta,z)\in (0,\infty)\times \R \times \R}{
  z>0}$.
\item $\lim_{T\downarrow 0} v(T,\eta,z) = \uV(\eta,z)$, where 
\begin{equation}
\label{equ:1E2F}
\begin{split}
  \uV(\eta,z):= \begin{cases}
 V(z)+z \vp(\eta), & z<z_c \\
 V(z_c) + z_c \vp(\eta)+(z-z_c) \inf\vp, & z \geq z_c,\\
 \end{cases} 
\end{split}
\end{equation}
and $z_c(\eta):=U'\big(\vp(\eta)-\inf\vp\big)\in (0,\infty]$, so that
$V'\big(z_c(\eta)\big) + \vp(\eta) = \inf \vp$.
\end{enumerate}
\end{proposition}
\begin{remark}
Proposition \ref{pro:fl-Brown} states that both
the primal and the dual value functions exhibit a \emph{facelift}
phenomenon:
\begin{enumerate}
\item In the primal case, the facelift ``cuts off'' a part of the
domain and leads to an effective initial condition for which the Inada
conditions fail. Indeed, $U'\big(x+\vp(\eta)\big)\not\to\infty$ as $x\to
-\inf\vp$ for all $\eta$, unless
$\vp$ is constant. On the other hand, as soon as $T>0$, we have $\pds{x}u(T,x,\eta)
\to \infty$, as $x\to -\inf \vp$, for all $\eta$.
\item 
The situation with the dual problem appears even more severe. Even though the
effective domain turns out to be exactly as expected, the limiting value 
$\uV$ of $v$ differs from 4' in the previous section. Indeed, unless
$\vp$ is constant,  we have  $\uV(z,\eta) < V(z)+z \vp(\eta)$ for $z>z_c(\eta)$.
\end{enumerate}
\end{remark}
One important consequence of the facelift is the following (the proof is
a combination of Remark \ref{rem:B1suf}, Corollary \ref{cor:nonattain},  Proposition \ref{pro:6DD4}):
\begin{proposition}
\label{pro:never-mart}
In the setting of the current section, let $(T,\eta_0)\in
(0,\infty)\times \R$  and a nonconstant
bounded and continuous function $\vp:\R\to\R$ be given. Then, 
for all large-enough $z>0$, the 
problem \eqref{equ:dual_value2} does not admit a minimizer in $\sM$.

If additionally $\E\big[z_c(\eta_0+W_T)\big]<\infty$, the previous statement holds
for all $z>0$. 
\end{proposition} 

Proposition \ref{pro:never-mart} shows that, even in the simplest of
     incomplete con\-ti\-nu\-ous-time financial models, the set of countably
     additive martingale measures $\sM$ is not big enough to  host the dual
     optimizer, as soon as the random endowment is unspanned
     (nonreplicable). A suitable relaxation (e.g., to the set of
     finitely-additive martingale measures) is therefore truly needed. 

From \cite{CviSchWan01} we know that the problem \eqref{equ:dual_value2}
always  admits a minimizer $\hat{\Q}$ in the weak$^*$-closure of $\sM$ in
ba$(\P)$. When $\hat{\Q} \notin \sM$ both components in
the Yosida-Hewitt decomposition $\hat{\Q} = \hat{\Q}^r + \hat{\Q}^s$ are
non-trivial. This follows because $V'(0) = -\infty$ forces
$\frac{d\hat{\Q}^r}{d\P} >0$, a.s.  However, closed-form expressions for $\hat{\Q}^r$ and $\hat{\Q}^s$ in the setting of 
     Proposition \ref{pro:never-mart} remain unavailable.

\section{A General Market Model}
\label{sec:General}

We start by describing a general semimartingale financial model which will
serve as the setting for our (abstract) result. It is built on
a filtered probability space
$(\Omega,\sF,\prfo{\sF_t},\P)$ which 
satisfies the usual conditions of right-continuity and completeness; we
assume, in addition, that $\sF_0$ is $\P$-trivial.  The choice of the
constant $1$ as the time horizon is arbitrary; it simply indicates that only
the values of the ingredients in a neighborhood of $0$ are of interest.

\subsection{The asset-price model} \label{sse:model} Let $\prfo{S_t}$ be an
$\prfo{\sF_t}$-adapted, RCLL semimartingale which satisfies the following
assumption (see \cite{DelSch98} for the definition and an in-depth
discussion of the concept of $\sigma$-martingale):
\begin{asmlistA}
\item \label{ite:NFLVR}
The set of  $\sigma$-martingale measures  
\[ \smes := \sets{ \QQ\sim \PP }{ S\text{ is a $\sigma$-
martingale under $\QQ$} } \] is nonempty.
\end{asmlistA}
As shown in \cite{DelSch98} (Theorem 1.1., p.~215), the assumption 
\ref{ite:NFLVR} is equivalent to the 
no-arbitrage condition
NFLVR (see \cite{DelSch98} for the details).

In the context of utility-maximization, it is easier to use a mild modification
$\sM$ of the set $\smes$, which is defined as follows.
Let the \define{admissible set} 
$\sA$ consists of all $\FFF$-predictable
  $S$-integrable processes $\pi$ such that $\int_0^{\cdot} \pi_u\, dS_u$
  is a.s.~uniformly bounded from below. We also define the set of gains
  processes $\sX$ by 
  \[\sX := \Bsets{ \int_0^{\cdot} \pi_u\, dS_u}{\pi\in \sA}.\]
Thanks to the $\sigma$-martingale property, each $X\in\sX$ is a
$\QQ$-local martingale (and therefore a supermartingale) for any $\QQ\in \smes$. Therefore, the set
\[ \sM := \sets{\QQ\sim\PP}{ X\text{ is a } \QQ\text{-supermartingale, for
each } X\in\sX}\]
includes $\smes$. The difference is often not very significant, since
Proposition 4.7, p.~239 in \cite{DelSch98} states that $\smes$ is dense in
$\sM$, in the total-variation norm.

\subsection{The utility function and its dual}
\label{sse:utility}
Let $U$ be a reasonably elastic utility function, i.e., a function
$U:(0,\infty)\to\R$ with the following properties:
\begin{asmlistA}
\item \label{ite:U}
$\left\{\begin{minipage}{0.85\textwidth}
\begin{enumerate}
\item $U$ is strictly concave, $C^1$, and strictly increasing on
$(0,\infty)$,
\item $\lim_{x\downto 0} U'(x) = +\infty$ and $\lim_{x\to\infty}
U'(x)=0$. 
\item \label{ite:AE} $\lim_{x\to\infty} \tfrac{x U'(x)}{U(x)}<1$, if $\sup_{x>0} U(x)>0$. 
\end{enumerate}
\end{minipage} \right.$
\end{asmlistA}
We extend $U$ to the negative axis by $U(x) :=-\infty$ for $x<0$ and $U(0) := \lim_{x\downto 0} U(x)$. $U$'s conjugate (dual utility function) $V:(0,\infty)\to (-\infty,\infty)$ is defined by
        \[ V(z) := \sup_{x\in \R} \Big( U(x) - xz \Big),\quad \efor z> 0.\]
        
\subsection{Value functions} 
  Given a bounded adapted and RCLL process $\prfo{\vp_t}$, 
for $x\in \R$ and $T\in [0,1]$ we set
  \begin{equation}
  \label{equ:def-sU}
    u(T,x) := \sup_{\pi\in\sA}
    \ee{ U \Big( x+\int_0^T \pi_u\, dS_u+\vp_T\Big)},
  \end{equation}
  with the usual convention that
  $\ee{\xi}=-\infty$, as soon as $\ee{\xi^-}=\infty$. 
We call $u$ the \define{(primal) value function} and note that
$u(0,x)=U(x+\vp_0)$, for all $x\in\R$. 
\begin{remark}
\label{rem:ED4}
 As in Section \ref{sec:illustrative}, 
 we use the (slightly nonstandard) notation $T$ for the time-variable to
 stress the fact that in our principal interpretation it plays the role of
 the \emph{time to go}.
 This is also done to avoid the possible confusion with the usual interpertation
 of the parameter $t$ as the current time, with the time-to-go being given
 by $T-t$. We continue using the variable $t$ as the generic ``dummy'' time parameter
 for stochastic processes. 
\end{remark}

For $\QQ\in\sM$, we let  the density process
$\prf{Z^{\QQ}_t}$ be the RCLL version of 
\[ Z^{\QQ}_t := \EE[ \tRN{\QQ}{\PP}|\sF_t],\, t \in [0,1].\]
We also set $\sZ := \sets{Z^{\QQ}}{\QQ\in\sM}$ and define
  \define{dual value function} $v:[0,1]
\times (0,\infty)\to (-\infty,\infty]$ by
 \begin{equation}
 \label{equ:4C6D}
 \begin{split}
   v(T,z) := \inf_{Z\in\sZ} \Bee{ V(zZ_T) + zZ_T \vp_T}.
 \end{split}
 \end{equation}
\begin{remark} \cite{CviSchWan01} show that the conjugate to the primal
value function \eqref{equ:def-sU} equals the expression on the
right-hand-side of \eqref{equ:4C6D} but with $\sZ$ replaced (in a suitable
manner - see Remark \ref{rem:infZ}) by its weak$^*$-closure in ba$(\P)$. 
Theorem
2.10, p.~675 in \cite{LarZit12} shows   that such a relaxation is not
necessary and that infimizing over $\sZ$ yields the same value function.
\end{remark} 
 
In order to have a nontrivial dual problem, we also ask for finiteness of
its value function on its entire domain:
\begin{asmlistA}
\item \label{ite:v-fin}
For all $T\leq 1$ and $z>0$,  we have
$v(T,z) <\infty$.
\end{asmlistA}
\begin{remark}
\label{rem:ae-v}
Thanks to the reasonable asymptotic elasticity condition \eqref{ite:AE} in Assumption \ref{ite:U}
(see \cite[Lemma 6.3, p.~944]{KraSch99} for details) 
the Assumption \ref{ite:v-fin} is equivalent to the existence of
$\QQ\in\sM$ such that $V^+(zZ^{\QQ}_1)\in\lone$, for \emph{some} $z>0$. In
that case, moreover, we have
$V(zZ^{\QQ}_1)\in\lone$ for \emph{all} $z>0$. 
\end{remark}
Remark \ref{rem:ae-v} above and 
the convexity of $V$ guarantee that the set 
\[ \sZ^V(z):=\sets{Z\in\sZ}{\EE[ V^+(zZ^{\QQ}_1)]<\infty},\]
is independent of $z>0$ (so we denote it simply by $\sZ^V$),  
nonempty, and 
enjoys the property that
\[ v(T,z) = \inf_{Z\in\sZ^V} \ee{ V(z Z_T)+ zZ_T \vp_T},
\text{
for all $T\in(0,1]$ and $z>0$.}\]
The set of all corresponding $\QQ\in\sM$ is denoted by $\sM^V$ and its
elements are referred to as \define{$V$-finite}. 
\subsection{The lower-hedging germ price}
\label{sse:lhg}
With $\sS_T$ denoting the set of all $[0,T]$-valued stopping times, we
define the \define{lower American germ price} of $\prfo{\vp}$ by 
  \[ \textstyle \Phi^A:=\lim_{T\downto 0} \Phi^A_T \ewhere \Phi^A_T :=
  \inf_{Z\in\sZ, \tau\in \sS_T} \EE[ Z_{\tau} \vp_{\tau}].\]
The \define{European} counterpart is defined as
  \[ \textstyle \Phi^{E}:=\limsup_{T\downto 0} \Phi^E_T \ewhere \Phi^E_T :=
  \inf_{Z\in\sZ} \EE[ Z_T \vp_{T}].\]
We assume that the two limits are equal:
\begin{asmlistA}
\item \label{ite:Phi} 
$\Phi^A = \Phi^E$,
\end{asmlistA}
and we denote the common value by $\Phi$ and call it the
\define{(lower-hedging) germ price}. 
\begin{remark} \ \label{rem:afterPhi}
\begin{enumerate}
\item Even though Assumption \ref{ite:Phi} is not as standard as,
e.g., \ref{ite:U} or \ref{ite:v-fin},  it is, in fact, quite mild and is 
satisfied in a wide variety of cases. Section \ref{sec:Phi} is
devoted to sufficient conditions and examples related to Assumption \ref{ite:Phi}.
We observe right away, however, that the following three properties follow
directly from it:
\[ \Phi \leq \vp_0,\ \Phi^A_T \nearrow \Phi \eand \Phi^E_T\to \Phi,\text{
  a.s., as } T\downto 0.\]
\item By the density of $\smes$ in $\sM$, referred to in Subsection 
\ref{sse:model}, 
the infima in the
definitions of $\Phi^A$ and $\Phi^E$ can be taken over the smaller set of densities
$Z$ of all $\sigma$-martingale measures $\smes$.
\item
Kabanov and Stricker 
have shown 
(see \cite{KabStr01a}, p.~140, Corollary 1.3)
that the infima in the definitions of lower hedging
prices can be taken over the set of all $V$-finite measures as long as $V$
satisfies a mild growth condition. Such a condition is satisfied in our
case thanks to the assumption of reasonable asymptotic elasticity in
\ref{ite:U}, part \eqref{ite:AE} and Lemma 6.3, p.~944 in \cite{KraSch99}, so
 \begin{equation}
 \label{equ:7094}
 \begin{split}
   \Phi^E_T = \inf_{Z\in\sZ^V}\EE[Z_T\vp_T] \text{ for all } T>0.
 \end{split}
 \end{equation}
This will be useful in the proof of Lemma \ref{lem:4938} below. 
\end{enumerate}
\end{remark}

\subsection{The form of the facelift and the main theorem}
Given two nonnegative constants $\vp,\psi$, let $z\mapsto \uV(z;\vp,\psi)$ denote the 
largest convex function below $z\mapsto V(z)+\vp z $ such that 
$\uV(z;\vp,\psi) - z \psi$ is nonincreasing. This function is given by
  \begin{equation}
  \label{equ:6A9E}
  \begin{split}
      \uV(z; \vp,\psi) = \sup_{x > -\psi} \Big( U(x+\vp) - xz \Big),
  \end{split}
  \end{equation}
  or, equivalently, by
 \begin{equation}
 \label{equ:314C}
 \begin{split}
    \uV(z;\vp,\psi) = \begin{cases}
   V(z) + \vp z,& z\leq z_c\\
   V(z_c) + \vp z_c + \psi (z-z_c), & z>z_c.
   \end{cases}
 \end{split}
 \end{equation}
Here $z_c$ is the (unique) solution to $V'(z_c) + \vp = \psi$ when it
exists, and $z_c=+\infty$, otherwise. 
The special case where $\vp=\vp_0$ and $\psi=\Phi$ 
appears in our main theorem below, so we give it its
own notation 
\begin{equation}
\label{equ:76ED}
\begin{split}
  \uV(z) := \uV(z;\vp_0,\Phi).
\end{split}
\end{equation}
We are now ready to state and prove the central result of this section - it
identifies explicitly the shape of the facelift in both the primal and the
dual problem. The proof is given in Subsection \ref{sse:proof} below.
\begin{theorem}
\label{thm:main}
 Under assumptions \emph{\ref{ite:NFLVR} - \ref{ite:Phi}}, we have
  \begin{equation}
  \label{equ:51E7}
  \begin{split}
     \lim_{T\downto 0} u(T,x) & = \begin{cases} U(x+\vp_0), & x > - \Phi \\
    -\infty,&   x<-\Phi  \end{cases}\\
    \lim_{T\downto 0} v(T,z) &= \uV(z),\text{ for all } z>0.
  \end{split}
  \end{equation}
\end{theorem}
\subsection{The true home of Theorem \ref{thm:main}}
\label{sse:markov}
One can argue that the natural home for
our facelifting result of Theorem \ref{thm:main} lies in a class of
interconnected optimization problems in a Markovian setting. Indeed, we would
\emph{not} like to adopt the somewhat unnatural 
interpretation of its result in the sense of the asymptotic behavior of the
dual value function as the time horizon shrinks to $0$ (with the random
endowment somehow depending on it). Rather, we would like to think of the
time as getting closer to the maturity, and the function $v$ as a section
of the entire time-dependent value function, in the spirit of the dynamic
programming principle. 
The way to pass from one framework to the other is rather
simple: when the dynamics of the underlying state process is
homogeneous, one can couple the problems corresponding to the same value of
the state, but with varying times, on the same probability space as
follows.

Let $\fF$ be a nonempty
Hausdorff LCCB (locally-compact with a countable
base), and therefore, Polish,  topological space - Euclidean or
discrete.  For  a nonempty
$G_{\delta}$ (in particular, open or closed) subset $\fS$ of $\R^d$, for some $d\in\N$, 
the product, $\fE=\fS\times \fF$ is Hausdorff LCCB and
Polish.  We work exclusively on the canonical space $\Omega=\de[0,\infty)$ consisting of
all $\fE$-valued RCLL (right-continuous with left limits) paths on
$[0,\infty)$, with the $\sigma$-algebra $\sF$  generated by all coordinate
maps. 
 
The coordinate process is denoted by $\eta$, and its 
components by \begin{enumerate}
\item $S=(S^1,\dots, S^d)$ - $\fS$-valued (modeling a risky
actively-traded asset), and 
\item $F$  - $\fF$-valued (modeling a non-traded factor).
\end{enumerate}
The ``physical'' dynamics of $\eta$ will be described
via a strong Markov family
$\petaeta$ of probability measures on $\de$. 
Let $\FFF^0$ be the (raw) filtration on $\de$, generated by the
coordinate maps, and let 
$\sF^{\eta}_t$ be the $\peta$-completion of $\sF^0_t$. 
Thanks to Blumenthal's $0$-$1$ law, $\FFF^{\eta}$ is right-continuous and 
satisfies the $(\peta)_{\eta\in\fE}$-usual conditions (see \cite[Chapter 3,
$\S$ 3, p.~102]{RevYor99} for details). 

To be able to use Theorem \ref{thm:main} under each $\peta$, we impose the
conditions \ref{ite:NFLVR}-\ref{ite:Phi} on each probability space
$(\Omega,\prfo{\sF^{\eta}_t},\sF,
\peta)$; the sets $\sM^{\eta}$ and $\sZ^{\eta}$ are simply the
$\eta$-parametrized versions of the eponymous objects defined earlier in
this section. Similarly, the admissible set depends on $\eta\in\fE$, and
the family is denoted by $\fam{\sA^{\eta}}{\eta\in\fE}$. 
We work with the utility function (and its dual) which satisfy the
conditions of \ref{ite:U}. 
Given a time-horizon $T\in (0,1]$ and $t\in [0,T]$ we define the primal value function 
\[ u(t,\eta, x) := \sup_{\pi\in\sA^{\eta}} 
\EE^{\eta} \Big[
U\Big( x+ \int_0^{T-t} \pi_u\, dS_u + \vp(\eta_{T-t}) \Big)\Big]\]
where $\vp$ is a bounded and continuous function on $\fE$.  Similarly, 
the dual value function is given by
\[ v(t,\eta,z) := \inf_{Z\in\sZ^{\eta}} 
\EE^{\eta} \Big[ V(zZ_{T-t}) + z Z_{T-t}\vp(\eta_{T-t}) \Big].\]
Under mild additional conditions on $S$ (it will, e.g., suffice that it is
either bounded from below or that its jumps are bounded from below), we
have the following version of the dynamic programming principle (see
Theorem 3.17 in \cite{Zit13}):
 \begin{equation}
 \label{equ:dpp}
 \begin{split}
    v(t,\eta,z) = \inf_{Z\in\sZ^{\eta}} \EE^{\eta}\left[ v(\tau, \eta_{\tau},
   zZ_{\tau})\right],\ v(T,\eta,z) = V(z) + z \vp(\eta), 
 \end{split}
 \end{equation}
for any 
  random time $\tau$ of the form $\tau=t+\sigma$, where $\sigma\in
  [0,T-t]$ is an $\FFF^{\eta}$-stopping time. 
It is also shown in \cite{Zit13} that the function $v$ is (jointly)
universally measurable, so that the expectation on the right-hand side of 
\eqref{equ:dpp} is well-defined. As shown in the last paragraph of
Subsection 3.4.~in \cite{Zit13}, the idea of the proof of Lemma
\ref{lem:54E8} below can be used to establish the dynamic programming principle for
the primal problem, as well. 

Equation \eqref{equ:dpp} 
often serves as an analytic description of the value
function. In continuous time it is usually infinitesimalised into a
PDE and studied, together with its terminal condition, as a nonlinear Cauchy
problem. As already mentioned in Section \ref{sec:illustrative}, in our case a
facelift (boundary-layer) phenomenon appears and this
terminal condition comes in a nonstandard form. Indeed, Theorem
\ref{thm:main} in the present setting becomes:
 \begin{equation}
 \label{equ:7E43}
 \begin{split}
   v(t,\eta,z) \to \underline{V}(z; \vp(\eta), \Phi(\eta)),
 \end{split}
 \end{equation}
where $\Phi(\eta)$ is as in Subsection \ref{sse:lhg}, with the dependence
on $\eta$ emphasized. We conjecture that \eqref{equ:dpp} and \eqref{equ:7E43}
suffice to characterize the value function $v$ in a wide class of models
(possibly via a PDE approach), but do not pursue this interesting question
in the present paper.

\subsection{A proof of Theorem \ref{thm:main}} 
\label{sse:proof}
We split the proof of our
main Theorem \ref{thm:main}
into lemmas and we start from a statement that allows us
to focus completely on the dual problem.
\begin{lemma}
\label{lem:54E8} Under assumptions \ref{ite:NFLVR}-\ref{ite:Phi}, the first
equality in \eqref{equ:51E7} follows from the second one.
\end{lemma}
\begin{proof}
Suppose that $\lim_{T\downto 0} v(T,z)=\uV(z)$, for all $z>0$.  
The conjugate relationship between the primal and the dual value functions
\[ u(T,x) = \inf_{z>0} \Big( v(T,z) + xz \Big) \quad \efor\quad  x\in\R,\quad  \quad T\in(0,1],\]
established in \cite{CviSchWan01} and further extended in
\cite{LarZit12}, allows us to 
apply the tools of classical convex analysis.
Indeed, the assumed pointwise convergence of the function $v$ transfers
directly to the convex conjugate in the interior of its effective domain (see Theorem 11.34, p.~500 in
\cite{RocWet98}). One only needs to check that the limiting function for
the primal value function in \eqref{equ:51E7} and the function $\uV$ are
convex conjugates of each other.
\end{proof}
We focus now exclusively on the dual problem and examine the asymptotic
behavior of the function $v$ in the large-$z$ regime:
\begin{lemma}
\label{lem:4938} Under assumptions \ref{ite:NFLVR}-\ref{ite:v-fin}, for all $T\in (0,1]$ the function $z\to v(T,z)$ is convex and 
 $$\lim_{z\to\infty} \oo{z} v(T,z) = \Phi^E_T.$$
\end{lemma}
\begin{proof} 
Convexity of $v(T,\cdot)$ follows from the convexity
of $V(T,\cdot)$ and $\sZ$. For the second statement, we fix $T\in(0,1]$,
pick an arbitrary $\eps>0$, and note that, for all $Z\in\sZ$, we have
 \begin{equation}
 \label{equ:6F7}
 \begin{split}
    \oo{z} \Bee{ V(z Z_T)+z Z_T \vp_T}
    \geq \oo{z}U(\eps)  + \EE[ Z_T (\vp_T-\eps)]
 \end{split}
 \end{equation}
 Passing to an infimum over all $Z\in\sZ^V$, and using the result in
 \eqref{equ:7094}, we get
\[ \tfrac{1}{z} v(T,z) \geq \oo{z} U(\eps) +  
\Phi^E_T - \eps,\text{ and so }
 \liminf_{z\to\infty} \oo{z} v(T,z) \geq  \Phi^E_T.\]
On the other hand, by the monotone convergence 
theorem we have
\[ \lim_{z\to\infty} \oo{z} \ee{ V(zZ_T)} = 0 \efor Z\in\sZ^V.\]
Therefore, for $Z\in\sZ$, we have
\begin{equation*}
\label{equ:344A}
\begin{split}
  \limsup_{z\to\infty} \oo{z} v(T,z) &\leq 
  \limsup_{z\to\infty} \oo{z} \ee{ V(zZ_T)} + \ee{ Z_T \vp_T} = 
  \EE[ Z_T \vp_T].
\end{split}
\end{equation*}
To complete the proof, it suffices to infimize over all $Z\in\sZ^V$.
\end{proof}
We define $v(0_+,z) := \liminf_{T\downto 0} v(T,z)$ and $v(0^+,z) := 
\limsup_{T\downto 0} v(T,z)$. 

\begin{lemma}
\label{lem:1D74} Under assumptions \ref{ite:NFLVR}-\ref{ite:Phi} 
$v(0^+,z)\leq \uV(z)$ for all $z>0$.
\end{lemma}
\begin{proof}
By Lemma \ref{lem:4938} above the function $z\mapsto v(T,z) - z \Phi^E_T $ is
convex and nonincreasing, for all $T\in(0,1]$. Therefore, so is the function $z\mapsto
v(0^+,z) - z\Phi$. Indeed, $\Phi = \lim_{T\downto 0} \Phi^E_T$ and both 
convexity and the nonincreasing property are preserved by
the limit superior operator. On the other hand, for $z>0$ and $Z\in\sZ^V$, the process
\[ t\mapsto V(z Z_t), \quad t\in [0,1],\]
is a uniformly integrable RCLL submartingale. Therefore, $\EE[V(zZ_T)] \to V(z)$,
as $T\to 0$. Since $\sZ^V$ is nonempty, we have
  \begin{equation*}
  \label{equ:3A92}
  \begin{split}
     v(0^+,z) &\leq \textstyle \limsup_{T\downto 0} \EE[ V(zZ_T)+zZ_T
     \vp_T] = V(z) + z \vp_0.
  \end{split}
  \end{equation*}
It remains to use the definition of $\uV$.
\end{proof}

\begin{lemma} \label{lem:2299} Under assumptions \ref{ite:NFLVR}-\ref{ite:Phi} 
$v(0_+,z)\geq \uV(z)$ for all $z>0$.
\end{lemma}
\begin{proof} 
For $T\in(0,1]$ and $t\in [0,1]$, we set 
\[ X_t = \esssup_{Z\in\sZ,\tau\in [t,T]}\EE[ -\tfrac{Z_T}{Z_t}
\vp_{\tau}|\sF_t] \efor t\leq T \eand X_t = X_T \efor t>T.\]
By Proposition 4.3, p.~467 in \cite{Kra96a}, 
$\prfo{X_t}$ admits a RCLL version
and the process $\prfo{Z_tX_t}$ is a
supermartingale for each $Z\in\sZ$. 
Also, we have
  $X_t+\vp_t \geq 0$, for all $t\leq T$ and
   $X_0  = - \Phi^A_T$. For $x>0$, Fenchel's inequality produces
\begin{equation*}
\label{equ:4DA2}
\begin{split}
  V(z Z_t) + zZ_t \vp_t \geq U( x +X_t+\vp_t)-zZ_t ( x +X_t), \as, \quad Z\in\sZ.
\end{split}
\end{equation*}
By taking expectation through we find
\begin{align*}
 \E[V(zZ_t)+zZ_t \vp_t] &\geq \EE[ U( x + X_t  + \vp_t)]  - zx -z\EE[ Z_t  X_t]\\ &
\geq \EE[ U( x + X_t  + \vp_t)]  - zx -zX_0\\ &
  = \EE[ U(  x  + X_t + \vp_t)] -zx + z \Phi^A_t.
\end{align*}
where the second inequality follows from the supermartingale property of $ZX$. Since $x>0$ we can use Fatou's Lemma to see 
\begin{equation*}
\label{equ:5219}
\begin{split}
v(0+,z) &\geq \liminf_{t\downto 0} \EE[ U( x   + X_t + \vp_t)] - z(x-\Phi)\\
&\geq \EE[ U( x  +\liminf_{t\downto 0} (X_t + \vp_t))] - z ( x -\Phi)\\ &= 
 U( x  + \vp_0 - \Phi^A_T) - z ( x  -\Phi),
\end{split}
\end{equation*}
where the last equality follows front the right-continuity of $X$'s and $\varphi$'s paths. It remains to let $T\downto 0$ and then maximize over all $x>0$. 
\end{proof}

\section{A modified objective}
\label{sec:mod}
Our next result states that the seemingly 
local effect of a facelift is sometimes felt far away from it, as well. We adopt the
setting of Section \ref{sec:General}, with assumptions
\ref{ite:NFLVR}-\ref{ite:v-fin} in place, but do not assume \ref{ite:Phi}. Since the results in this section are not asymptotic in
nature, we chose and fix a time horizon $T>0$ and replace the time-set
$[0,1]$ from Section \ref{sec:General} by the generic $[0,T]$.

As a preparation for our result on the modified objective we define the set
\[ \sC:=\Bsets{ x+\textstyle\int_0^T \pi_u\, dS_u}{x\in\R, \pi\in\sA}.\]
The following property for the variable $\vp_T$ will be crucial in the
sequel:
\begin{asmlistB}
 \item There exists a random variable $\uvp_T\in\sC$ such that
 $X+\uvp_T\geq 0$, a.s., whenever
 $X\in \sC$ and  $X+\vp_T\geq 0$, a.s. 
 \label{ite:smalfac}
\end{asmlistB}
\begin{remark}
\label{rem:B1suf}
One can construct one-period examples on a three-element probability
space where \ref{ite:smalfac} fails. There are, nevertheless, plenty of
cases when it always holds. 
 For example, in \cite{BroCviSon98}, 
\ref{ite:smalfac} is shown to hold in a related problem.
In particular in the setting of Section \ref{sec:illustrative}, we have 
$$
x+ \int_0^t \pi_u dS_u \ge \E^\QQ\left[ x+\int_0^T\pi_u dS_u|\sF_t\right]\ge -\E^\QQ[\varphi(\eta_0 + W_T) |\sF_t].
$$
By optimizing over $\Q\in\sM$ we then find
$$
x+ \int_0^t \pi_u dS_u \ge -\inf\vp.
$$
Consequently, in the setting of Section \ref{sec:illustrative}, we have $\uvp_T=\inf \vp$. 

\end{remark}
\begin{theorem} 
\label{thm:main2}
Suppose that assumptions \ref{ite:NFLVR}-\ref{ite:v-fin} and
\ref{ite:smalfac} hold.
Then, 
 for all $z \in  (0,\infty)$ and $T\in(0,1]$ we have the representation
  \begin{equation}
  \label{equ:5375}
  \begin{split}
      v(T,z) = \inf_{Z\in\sZ} \EE[ \uV(zZ_T; \vp_T,\uvp_T)].
  \end{split}
  \end{equation}
\end{theorem}
\begin{proof}
 Let $(T,z)\mapsto \uv(T,z)$ 
 denote the function defined by the right-hand side of
\eqref{equ:5375}. Since $\uV(z; \vp_T,\uvp_T)\leq V(z) + z \vp_T$, for all $z>0$,
a.s., we clearly have $\uv\leq v$. To prove
the converse inequality, we pick
$x\in\R$ and  $\pi\in\sA$,  with $\EE[ U(x+
\ips +\vp_T)]>-\infty$. That implies
\[ x+\ips+\vp_T \geq 0, \as\]
Therefore, there exists $\uvp_T\in\sC$ such that
\[ x+\ips+\uvp_T \geq 0.\]
This produces 
$$
\EE[ U(x+\ips+\vp_T)] = \EE[ \underline{U}\big(x+\ips+ \vp_T; \uvp_T\big)],
$$ 
where we have introduced
\[ 
  \underline{U}(x;\uvp) := \begin{cases} U(x), & x > -
  \uvp,\\
-\infty, & \text{ otherwise.}\end{cases}
\]
We have
 \begin{equation*}
 \begin{split}
   \sup_{x\in\R} \Big( \underline{U}(x+\vp) - xz \Big) &=
     \sup_{x> -\uvp} \Big( U(x+\vp) - xz \Big) 
    = \uV(z; \vp,\uvp),
 \end{split}
 \end{equation*}
which, together with the supermartingale property of $Z_t (x+\int_0^t \pi_u
dS_u)$ for each  $Z\in\sZ$ produces
 \begin{equation*}
 \label{equ:3B93}
 \begin{split}
    \EE[ U(x+\ips + \vp_T)] - xz\leq \EE[ \underline{V}(zZ_T;\vp_T,\uvp_T)
   ] ,
 \end{split}
 \end{equation*}
for all $z>0$. This,  in turn, implies that
\[ u(T,x) - xz \leq \uv(T,z).\]
The claim now follows by using the conjugacy of the primal and dual value
functions, as established in \cite{CviSchWan01} and extended in
\cite{LarZit12} (see Remark \ref{rem:infZ} above for  details.) 
\end{proof}
The result of Theorem \ref{thm:main2} has an interesting consequence:
\begin{corollary}
\label{cor:nonattain}
 Suppose that the conditions of Theorem \ref{thm:main2} hold and that $
 \PP[ \uvp_T \ne \vp_T]>0$. Then, the dual problem 
 \eqref{equ:4C6D} \emph{does not} admit a minimizer $Z\in\sZ$, for all
 $z>0$, large enough. 
\end{corollary}
\begin{proof}
Since $\uvp_T\ne \vp_T$, with positive probability, there exists a constant
$z_0\in (0,\infty)$ such that
 \begin{equation}
 \label{equ:4616}
 \begin{split}
   \PP[ V(z)+z \vp_T >  \uV(z; \vp_T,\uvp_T) ] >0, \text{ for } z\geq z_0.
 \end{split}
 \end{equation}
Additionally, we have the trivial inequality
$V(z)+z \vp_T  \geq  \uV(z; \vp_T,\uvp_T)$, a.s., for all $z\in(0,\infty)$.  
Suppose, now, that $\hZ = \hZ(z)$ is the dual minimizer, i.e., the minimizer in
\eqref{equ:4C6D}, corresponding to $z\geq
z_0$. By Theorem \ref{thm:main2}, it must also be a minimizer for 
the right-hand side of \eqref{equ:5375}, and it must have the property that
\[ V(zZ_T) + z\hZ_T \vp_T = \uV(z\hZ_T; \vp_T,\uvp_T), \as\]
The inequality in \eqref{equ:4616}, however, implies that
\[ \PP[ z\hZ_T< z_0]=1, \]
which is in contradiction with $z\geq z_0$ and $\EE[ \hZ_T]=1$. 
\end{proof}
In the model of Section \ref{sec:illustrative} one can improve on 
Corollary \ref{cor:nonattain} and show non-attainment
for any $z>0$, provided that $\vp$ does not stay ''too close'' to its
  minimum:
\begin{proposition}
\label{pro:6DD4}
In the setting of Section \ref{sec:illustrative}, we assume that
 \begin{equation}
 \label{equ:666E}
 \begin{split}
   \EE[ U'( \vp(\eta_0+W_T) - \inf\vp)] < \infty.
 \end{split}
 \end{equation}
Then the dual problem 
\eqref{equ:dual_value2} at $(T,\eta_0)$  does not admit a minimizer in
$\sZ$ for any $z>0$.  
\end{proposition}
\begin{proof}
Given \eqref{equ:666E}, we  assume that there exists $z>0$ and
$\hat{Z}\in\sZ$ which attains the infimum in \eqref{equ:4C6D}. As in the
proof of Corollary \ref{cor:nonattain}, this implies that $z\hat{Z}_T\leq 
Y$, a.s.,  where $Y:= U'\big(\vp(\eta_0+W_T\big) -
\inf\vp)$. Thanks to the special structure of the set $\sZ$ in the model of
Section \ref{sec:illustrative}, there exists a predictable and $W$-integrable 
process $\prf{\hnu_t}$ such that
\[ \hat{Z}_T = \sE(-\lambda \cdot B)_T\hat{H}_T, \quad \hat{H}_t:= \EN( \hnu\cdot W)_t.\]
We define the filtration $\prf{\sG_t}$ as the usual augmentation of  
\[ \sG^{raw}_t:=\sigma( B_u, W_s; u\leq T, s\leq t),\ t\in [0,T].\]
The process $W$ is a $\sG$-Brownian motion and $\hat{\nu}$ is
$\sG$-predictable, so $\hat{H_t}$ is a $\sG$-local martingale and, in
particular, we have $\EE[ \hat{H}_T|\sG_0]\leq 1$. 
Therefore, 
 \begin{equation*}
 \label{equ:4F5E}
 \begin{split}
   e^{\ld B_T - \tot \ld^2 T} 
   \geq 
   e^{\ld B_T - \tot \ld^2 T} 
   \EE[ \hat{H}_T|\sG_0] = \EE[ \hat{Z}_T|\sG_0],
 \end{split}
 \end{equation*}
where the inequality is, in fact, an a.s.-equality since  
both sides have expectation $1$. Using the fact that $Y$ is independent of
$\sG_0$, we conclude that
\[ e^{\ld B_T - \tot \ld^2 T} \leq\tfrac1z \EE[ Y] <\infty,\as,\]
which is a contradiction with the fact (derived from the assumption that
$\mu\ne 0$) that the distribution of 
left-hand side has support $(0,\infty)$.
\end{proof}
\section{Sufficient conditions for \ref{ite:Phi}}
\label{sec:Phi}
Condition \ref{ite:Phi} in Section \ref{sec:General} plays a major role
in the proof of Theorem \ref{thm:main} and guarantees that the process
$\prfo{\vp_t}$ does not oscillate to much as $t\downto 0$. Clearly, it (or
a version of it) must be imposed - indeed, the very form of the facelift
depends on the value (and existence) of the limiting germ price $\Phi$. 
We
present here two sufficient conditions for its validity which apply to a
wide variety of situations often encountered in mathematical finance. 

\subsection{Complete markets}
In the case of a complete market we have:
\begin{proposition}
If $\sM=\set{\QQ}$, for some $\QQ\sim\PP$, then $\ref{ite:Phi}$ holds with
 \begin{equation}
 \label{equ:125F}
 \begin{split}
   \Phi^E=\Phi^A=\vp_0.
 \end{split}
 \end{equation}
\end{proposition}
\begin{proof}
It suffices to note that, by the
dominated convergence theorem and the RCLL assumption, we have 
\[ \EE^{\QQ}\Big[ \inf_{t\in [0,T]} \vp_t \Big] \to \vp_0 \eand \EE^{\QQ}[
\vp_T] \to \vp_0,\text{ as } T\downto 0.\qedhere\]
\end{proof}
\subsection{Sufficient controllability}
Our second sufficient condition assumes
that there exists a process
$\prfo{\eta_t}$ with values in some topological space $\fE$ such that
$\vp_t= \vp(\eta_t)$, $t\in [0,1]$ for some continuous and bounded function
$\vp:\fE\to \R$.

We start with a general condition - phrased as a lemma -  which,
heuristically, says that \ref{ite:Phi} holds if $\prfo{\eta_t}$ can be well-controlled towards
any point in $\fE$, within any positive amount of time. To state it, 
for each $T\in (0,1]$ we define the set
$\sD_T$ of $\QQ$-distributions of $\eta_t$, as $\QQ$ ranges through
$\sM$ and $t\in (0,T]$. 
\begin{lemma}
\label{lem:on} 
Suppose that $\vp$ is a bounded and continuous function and that 
for each $\eta\in\fE$ and each $T>0$ 
there exists a sequence $\seq{\mu}$ in $\sD_T$ such that 
$\mu_n \to \delta_{\eta}$, weakly.
Then $\ref{ite:Phi}$ holds with
 \begin{equation}
 \label{equ:7C07}
 \begin{split}
   \Phi^E=\Phi^A=\inf\vp.
 \end{split}
 \end{equation}
\end{lemma}
\begin{proof}
Since $\vp$ is continuous and bounded, the assumptions imply that
\[ \inf_{\mu\in\sD_T} \int \vp(\eta)\, d\mu(\eta)= \inf \vp.\]
Therefore, $\Phi^E_T = \inf \vp$ for $T\in(0,1]$. Recalling that, by
construction, $\inf\vp \leq \Phi^A_T\leq \Phi^E_T$ 
for all $T\in(0,1]$, we conclude that
\eqref{equ:7C07} holds. 
\end{proof}
Next, we describe a large class of models with $\fE:=\R^d$ to which Lemma
\ref{lem:on} applies. We start by fixing a filtered probability space with a filtration
$\FFF$ satisfying the usual conditions.  Let $\sS^g$ denote the 
set of all $\R^d$-valued semimartingales  
$R$ with $R_0=0$ for which there exists
\begin{enumerate}
\item a semimartingale decomposition $R=M+F$ into a local martingale $M$
and a finite-variation process $F$, and
\item a (deterministic) function
$g_R:[0,1]\to [0,\infty)$ with $\lim_{T\downto 0} g_R(T)= 0$, 
\end{enumerate}
such that, with $\abs{F}$ denoting the total-variation process of $F$ and
$[M,M]$ the quadratic-variation process of $M$, we have
\[ \abs{F}_T+[ M,M ]_T\leq g_R(T), \as\text{ for all } T\in [0,1].\]
We note that, a posteriori, membership in $\sS^g$ immediately makes any
semimartingale special and we can (and do) talk about its unique
semimartingale decomposition without ambiguity.
\begin{remark}

\label{rem:635A}
An example of an element in the class $\sS$ is a process of the form
 \[ R_t= \int_0^t \alpha_u\, du + \int_0^t \beta_u\, dB_u+\int_0^t
 \gamma_u\, dN_u,\  t\in [0,1],\]
 where $\alpha$, $\beta$  and $\gamma$ are uniformly bounded predictable processes
 valued in, respectively $\R^d$,   $\R^{d\times d}$  and $\R^d$; 
 $B$ is a $d$-dimensional Brownian motion while $N$ is a $d$-dimensional
 Poisson process. 
\end{remark}
The class $\sS^g$ is important in our setting because it admits moment
estimates uniform over all equivalent measure changes that preserve the
semimartingale decomposition. The next result follows directly from the
Burkholder-Davis-Gundy inequalities (see Theorem 48, p.~193 in \cite{Pro04}), and we skip the proof:
\begin{lemma} 
\label{lem:BDG}
For each $R\in \sS^g$ with the semimartingale decomposition $R=M+F$
there exists a function $h_R:[0,1]\to (0,\infty)$ with $h_R(t)
\to 0$, as $t\downto 0$ such that
\[ \EE^{\QQ} [\abs{R_{t}}]   \leq h_R(t)\text{ for all } t\in [0,1],\]
for any 
$\QQ\sim\PP$ such that $M$ is a
$\QQ$-local martingale. 
\end{lemma}
\begin{theorem}
\label{thm:suff}
Suppose that the $\R^m$-valued 
process $S$ and the $\R^d$-valued factor process $\eta$
are semimartingales which satisfy the following assumptions:
\begin{enumerate}
\item there exists a $\P$-equivalent measure $\Q^0$, such that $S$ is local martingale,
\item the process $\prfo{\eta_t}$ is of the form  
\[ \eta_t = \eta_0 + \int_0^t \beta_u\, dW_u+R_t,\ t\in [0,1],\]
where 
$W$ is a Brownian motion strongly orthogonal to $S$ and to $Z^0$ (the
density process of $\QQ^0$ w.r.t. $\PP$), $\beta$ is a
bounded predictable process whose absolute value is bounded away from $0$,
and
$R\in \sS^g$ with the semimartingale decomposition $R=M+F$ where $M$ is
strongly orthogonal to $W$. 
\end{enumerate}
Then the condition \ref{ite:Phi} holds and 
   $\Phi^E=\Phi^A=\inf \vp$.
\end{theorem}
\begin{proof}
We start by constructing a large-enough subfamily of the  
family of local martingale measures. 
For a bounded, predictable
process $\prfo{\nu_t}$ we define the following two processes:
\begin{align*}
H^{\nu}_t := \EN(\nu\cdot W)_t,\quad Z^{\nu}_t := Z^0_t H^{\nu}_t,\quad t\in [0,1].
\end{align*}
The strong orthogonality between $W$ and $Z^0$ and the continuity of $W$ ensure  that $[W,Z^0]\equiv 0$, hence, $Z^\nu$ is a local martingale. To see that $Z^\nu$ is a martingale we note that
$$
\E[Z^\nu_T] = \E[ Z_T^0 H^\nu_T] = \E^{\QQ^0}[H^\nu_T].
$$
Since $W$ remains a Brownian motion under $\Q^0$ and since $\nu$ is bounded we can use Novikov's condition to see $\E^{\QQ^0}[H^\nu_T]=1$ from which the martingale property follows. We can then define $\frac{d\Q^{\nu}}{d\P} := Z^\nu_T$.

We fix a constant $\eta\in\R$. For $n\in\N$ we define the bounded process
\[ \nu^n_t := \begin{cases}
\tfrac{n (\eta-\eta_0)}{\beta_t}, & t\leq 1/n, \\ 0, & \text{ otherwise.} 
\end{cases}\]
Then $\eta_0+\int_0^t \beta_u \nu^n_u\, du = \eta$, for $t\geq 1/n$, and
\[ \eta_{1/n} = \eta_0 + R_{1/n} + \int_0^{1/n} \beta_u\, dW_u= \eta + R_{1/n} + \int_0^{1/n} \beta_u\, dW^n_u,\]
where $W^n_t := W_t - \int_0^t \nu^n_u\, du$. 
Thanks to the orthogonality assumption, 
the local-martingale part $M$ of $R$ is a $\QQ^{\nu^n}$-local martingale, and so
$\QQ^{\nu_n}$ and $R$ satisfy 
the conditions of Lemma \ref{lem:BDG}. Moreover, $W^n$ is a
$\QQ^{\nu_n}$-Brownian motion, so
\[ \EE^{\QQ^{\nu^n}}\Big[ \babs{\eta_{1/n}-\eta}\Big]\leq C \left( 
\EE^{\QQ^{\nu^n}}\Big[
\babs{R_{1/n}}\Big]+\EE^{\QQ^{\nu^n}}\Big[\babs{\textstyle\int_0^{1/n}
\beta_u\, d\tW_u}\Big]\right),\]
for some constant $C$.
The right-hand side is bounded from above 
by a linear combination of $h_R(1/n)$,
 and $1/n$, so it converges to $0$ as $n\to\infty$. 
Therefore, we have
\[ \int_{\R^d} \abs{x-\eta}\, \mu_n(dx) \to 0,\text{ as } n\to\infty,\]
where $\mu_n$ denotes $\eta_{1/n}$'s distribution under $\QQ^{\nu^n}$.
Consequently, $\mu_n \toweak \delta_{\eta}$ and Lemma \ref{lem:on} can be applied. 
\end{proof}
\begin{remark}
A family of examples of models which satisfy Assumption (1) of
Theorem \ref{thm:suff}
is furnished by
processes of the form
 \begin{equation}
 \label{equ:5AB0}
 \begin{split}
   dS_t := \mu_t\, dt+\sigma_t\, dB_t+dJ_t,\quad S_0 \in\R,
 \end{split}
 \end{equation}
where $B$ is an $\FFF$-Brownian motion independent
of $W$, $J$ is an $\FFF$-local martingale (possibly with jumps) which is
strongly orthogonal to $B$, and $\mu$ and $\sigma$ are predictable
processes. To apply Theorem \ref{thm:suff} it suffices to note that the
process $Z_t^0:=\EN(- \mu/\sigma\cdot B)_t$ is a strictly positive
martingale whenever $\mu/\sigma$ is sufficiently integrable.
\end{remark}


\def\ocirc#1{\ifmmode\setbox0=\hbox{$#1$}\dimen0=\ht0 \advance\dimen0
  by1pt\rlap{\hbox to\wd0{\hss\raise\dimen0
  \hbox{\hskip.2em$\scriptscriptstyle\circ$}\hss}}#1\else {\accent"17 #1}\fi}
  \ifx \cprime \undefined \def \cprime {$\mathsurround=0pt '$}\fi\ifx \k
  \undefined \let \k = \c \fi\ifx \scr \undefined \let \scr = \cal \fi\ifx
  \soft tundefined \def \soft
  {\relax}\fi\def\ocirc#1{\ifmmode\setbox0=\hbox{$#1$}\dimen0=\ht0
  \advance\dimen0 by1pt\rlap{\hbox to\wd0{\hss\raise\dimen0
  \hbox{\hskip.2em$\scriptscriptstyle\circ$}\hss}}#1\else {\accent"17 #1}\fi}
  \ifx \cprime \undefined \def \cprime {$\mathsurround=0pt '$}\fi\ifx \k
  \undefined \let \k = \c \fi\ifx \scr \undefined \let \scr = \cal \fi\ifx
  \soft tundefined \def \soft {\relax}\fi
\providecommand{\bysame}{\leavevmode\hbox to3em{\hrulefill}\thinspace}
\providecommand{\MR}{\relax\ifhmode\unskip\space\fi MR }
\providecommand{\MRhref}[2]{%
  \href{http://www.ams.org/mathscinet-getitem?mr=#1}{#2}
}
\providecommand{\href}[2]{#2}

\end{document}